\documentclass{article}
\usepackage{amsmath, amssymb}
\usepackage{amsfonts}
\usepackage{appendix}
\usepackage{graphicx}
\usepackage{color}
\usepackage{url}
\usepackage{bm}

\newtheorem{fact}{Fact}

\newtheorem{theorem}{Theorem}

\newtheorem{corollary}{Corollary}

\newtheorem{definition}{Definition}

\newtheorem{remark}{Remark}

\newenvironment{proof}{{\noindent\it Proof}\quad}{\hfill $\square$\par}
\newtheorem{example}{Example}

\newcommand{\F}{\ensuremath{\mathbb F}}
\newcommand{\Z}{\ensuremath{\mathbb Z}}

\newcommand{\ls}[1]
    {\dimen0=\fontdimen6\the\font\lineskip=#1\dimen0
     \advance\lineskip.5\fontdimen5\the\font
     \advance\lineskip-\dimen0
     \lineskiplimit=0.9\lineskip
     \baselineskip=\lineskip
     \advance\baselineskip\dimen0
     \normallineskip\lineskip\normallineskiplimit\lineskiplimit
     \normalbaselineskip\baselineskip
     \ignorespaces}

\begin{document}

\bibliographystyle{abbrv}

\title{Constructions of complementary sequence sets and complete complementary codes by 2-level autocorrelation sequences and permutation polynomials}
\author{Zilong Wang$^1$ and Guang Gong$^2$, IEEE Fellow\\
	\small $^1$ State Key Laboratory of Integrated Service Networks, Xidian University \\[-0.8ex]
	\small Xi'an, 710071, China\\
	\small $^2$Department of Electrical and Computer Engineering, University of Waterloo \\
	\small Waterloo, Ontario N2L 3G1, Canada  \\
	\small\tt zlwang@xidian.edu.cn, ggong@uwaterloo.ca\\
}
	\maketitle

\thispagestyle{plain} \setcounter{page}{1}

\begin{abstract}
In this paper, a recent method to construct complementary sequence sets and complete complementary codes by Hadamard matrices is deeply studied. By taking the algebraic structure of Hadamard matrices into consideration, our main result determine the so-called $\delta$-linear terms and  $\delta$-quadratic terms.  As a first consequence, a powerful theory linking Golay complementary sets of $p$-ary ($p$ prime) sequences and the generalized Reed-Muller codes by Kasami et al. is developed.  These codes enjoy good error-correcting capability, tightly controlled PMEPR, and significantly extend the range of coding options for applications of OFDM using $p^n$ subcarriers. As another consequence, we make a previously unrecognized connection between the sequences in CSSs and CCCs and the sequence with 2-level autocorrelation, trace function and permutation polynomial (PP) over the finite fields, which provides an answer to the open problem proposed by Paterson and Tarokh in 2000.
\end{abstract}

{\bf Index Terms }Complementary sequence set, Complete complementary codes, Permutation polynomial, Generalized Reed-Muller codes, 2-level auto-correlation.

\ls{1.5}
\section{Introduction}

The concept of {\em Golay sequence pair} (GSP) was first  introduced by Golay \cite{Golay51, Golay61}, and it was  extended later  to {\em complementary sequence set} (CSS) for binary case \cite{Tseng72} and polyphase case
\cite{Sivaswamy78}, where the aperiodic autocorrelations
of all the sequences in a CSS are summed to zero except
at zero shift. The concept of {\em complete mutually orthogonal complementary set} (CMOCS) or {\em complete complementary code} (CCC) was  proposed in \cite{Suehiro88}, which can be regarded as a collection
of CSSs with the additional aperiodic cross-correlation
property. CSSs  and CCCs  have been
applied in diverse areas of digital communications, including channel measurement, synchronisation, spread spectrum communications
and power control for multi-carrier wireless transmission. In addition, CCCs are source to design zero correlation zone (ZCZ) sequences, which have been shown in \cite{ZCZ13,ZCZ06,Tang2010}.

Among these applications, orthogonal frequency division multiplexing (OFDM) has recently seen rising popularity in international standards including coming 5G cellular systems, and one attractive method for power control in OFDM \cite{Litsynbook} is coding across the subcarriers and selecting codewords with lower peak-to-mean envelope power ratio (PMEPR). These codewords can be a linear code or codewords drawn from cosets of a linear code \cite{Tarokh00,PT2000}.
It has been shown in \cite{Popovic91} that the use of sequences in CSSs as codewords results in OFDM signals
with low PMEPR.

An effective way of combining the coding approach  and the use of Golay sequences was established by Davis and Jedwab \cite{DavisJedwab99} by showed that
the  Golay  sequences can be obtained from specific second-order cosets of the  first-order generalized Reed-Muller (GRM) code \cite{MacWilliams}.
Form then on, a large volume of research works  on the constructions of CSSs  \cite{Paterson00,Paterson02,Schmidt06,Schmidt07,Chen06, Stinchcombe, Wu2016} and CCCs \cite{Chen08,Liu14,CCC} have been done along this line that the sequences are all described by the generalized Boolean functions (GBFs) and fall into cosets of the GRM code.

In addition to the above direct constructions of CSSs and CCCs based on GBFs, there
exists a recursive approach to construct CSSs and CCCs based on  para-unitary (PU) matrices, such as \cite{Tseng72,Marziani,Suehiro88, SETA2016,BudSpas,BudIT,TSP2018}.
Recently, we  proposed a framework \cite{CCA} on the construction of PU matrices by Hadamard matrices,  which established a connection between the  aforementioned GBF-based constructions and PU-based constructions, and a great number of the new CSSs and CCCs are constructed. The key theoretical result of this method is to  extract functions from the so called {\em generalized seed matrices}, from which the $q$-ary  sequences of length $N^m$ in CSSs and CCCs of size $N$ can be represented by the {\em $\delta$-linear terms} and {\em $\delta$-quadratic terms}, where the word {\em linear} and {\em quadratic} are with respect to the  Kronecker-delta functions from $\Z_N$ to $
\Z_q$. We have shown some realized constructions by such a method in \cite{CCA} by Bruce-force computation, which is very heavy even for $q=N=4$.

In this paper, by taking the algebraic structure of the {\em Butson-type}  Hadamard (BH) matrix into consideration, the structure of the $\delta$-linear terms and  $\delta$-quadratic terms is explicit given, which avoid the heavy computation from the basis of  Kronecker-delta functions. In particular, we make an previously unrecognized connection between the sequences in CSSs and CCCs and the sequence with 2-level autocorrelation, trace function and permutation polynomial (PP) over the finite fields.

By taking the Discrete Fourier transform (DFT) matrix of order $N$ as a BH matrix, we obtained that the $\delta$-quadratic terms can be represented by the product of any two permutation functions over $\Z_N$. For $N=2^n$ being a power of $2$, such  permutation functions can be realized by the bijective GBFs. For $N=p$ prime, such  permutation functions can be realized by PPs over finite field $\F_p$.  A powerful theory linking $p$-ary sequence of length $p^m$ in CSSs of size $p$ and the GRM codes proposed by  Kasami,  Lin  and Peterson \cite{Kasami68} is developed. The reader should note that there are two types of  nonbinary generalizations of the classical RM codes. The first type directly generalize alphabets of classical RM codes from binary to $q$-ary, which have been made a connection with sequences in CSSs of size $2^n$, such as \cite{DavisJedwab99, Paterson00}. Here the GRM code is the another type of generalization proposed by Kasami et al. in 1968. For $p=2$, the constructed sequences agree with the binary Golay sequences \cite{DavisJedwab99}. For $p$ odd prime, the $\delta$-linear terms, which contains $p^{m(p-1)+1}$ elements, can be treated as a linear subcode of the $(p-1)$th order second type  GRM codes, and we explicitly identify $\frac{1}{2}m!\left((p-1)!\right)^{m-1}\left((p-2)!\right)^{m-1}$ cosets within this specified linear subcode from the $\delta$-quadratic terms. Moreover, their error correction capability is given by showing that Hamming distance of any two distinct sequence is  at least $3p^{m-2}$.

By taking the Hadamard matrix of order $p^n$  over the finite field $\F_{p^n}$ as a BH matrix, we obtained that the $\delta$-quadratic terms can be represented by the trace function of  product of any two PPs over $\F_{p^n}$. Note that there have been numerous books and papers on the study of PPs over finite fields. A wealth of results covering different periods in the development of this active area can be found in \cite{Lausch, Lidl, Mullen, Hou}. Nevertheless, this is the first time to use PPs over the finite field to construct CSSs and CCCs.

Every $p$-ary ($p$ prime) sequence with {\em 2-level auto correlation}  of period $p^n-1$ determine a BH matrix of order $p^n$, where each entry can be represented by the trace representation of the sequence. Then the $\delta$-quadratic terms can be represented by the trace representation of the 2-level auto correlation sequence with the product of any two PPs over $\F_{p^n}$ as variable.  Note that there is a large volume research on the construction of the sequences with 2-level autocorrelation, which correspond to the cyclic Hadamard difference sets,  such as $m$ sequences, GMW sequences, WG sequences and so on. A collection of the results in this area can be found in \cite{Gong-book}.
We showed that $m$ sequences yield the same result as Hadamard matrix over the finite field $\F_{p^n}$, and other 2-level autocorrelation sequences produce new CSSs and CCCs.
Nevertheless, this is the first time to use the trace representation of any sequence with 2-level auto correlation to construct CSSs and CCCs.

The coset representatives of sequences proposed in this paper are presented by the trace function, so the results in this paper provide an answer to 
the open problem proposed by K.G. Paterson and V. Tarokh \cite{PT2000} in 2000:
{\em It may be possible to obtain significant reductions in PMEPR by using such offsets, and we leave the analytical determination of good offsets for
trace codes as a difficult open problem.}

The rest of our paper is organized as follows. In the next section, we introduce the definition and notations, revisit the framework  on the construction of PU matrices by Hadamard matrices in \cite{CCA}. In Section 3,  we study on the $\delta$-linear terms and $\delta$-quadratic terms. In Section 4, we take DFT matrices as BH matrices in $\delta$-quadratic terms, and link $p$-ary sequence of length $p^m$ in CSSs of size $p$ with the GRM codes proposed in \cite{Kasami68}.
In Section 5, we take Hadamard matrices over the finite filed as BH matrices in $\delta$-quadratic terms, and construct CSSs and CCCs by PPs and trace function over finite fields. Section 6 make a connection of the constructions of CSSs and CCCs with the sequences with 2-level autocorrelation. We conclude the paper in Section 7.

\section{Preliminaries}

In this section, we introduce some basic definitions and notations of CSSs, CCCs, and two different type generalizations of the classical   Reed-Muller codes. A  framework on the construction of CSSs and CCCs by Hadamard matrices \cite{CCA} are also revisited.

\subsection{Sequences, CSS and CCC}

Let $\Z_p$ be a residue class ring modulo $p$. If $p$ is a prime, we use the finite filed $\F_p$ instead of the ring $\Z_p$.
A $q$-ary sequence $\bm{f}$ of length $L$ is defined as
$$\bm{f}=(f(0), f(1),\cdots, f(L-1))$$
for each entry $f(t)\in \mathbb{Z}_q$ ($t\in \Z_L$).

\begin{definition}\label{correlation-def}
For two $q$-ary sequences $\bm{f}_1$ and $\bm{f}_2$ of length $L$, the {\em aperiodic cross-correlation} of $\bm{f}_1$ and $\bm{f}_2$ at shift $\tau$ ($-L<\tau< L$) is defined by
$$C_{\bm{f}_1,\bm{f}_2}(\tau)=
\left\{
\begin{aligned}
&\sum_{t=0}^{L-1-\tau}{\omega^{\bm{f}_1(t+\tau)-\bm{f}_2(t)}},\
0\leq \tau<L,\\
&\sum_{t=0}^{L-1+\tau}{\omega^{\bm{f}_1(t)-\bm{f}_2(t-\tau)}},\
-L<\tau<0,
\end{aligned} \right.
$$
where $\omega$ is a $q$th root of unity.
If $\bm{f}_1=\bm{f}_2=\bm{f}$, the {\em aperiodic autocorrelation} of  sequence $\bm{f}$ at shift $\tau$ is denoted by
$$C_{\bm{f}}(\tau)=C_{\bm{f},\bm{f}}(\tau).$$
\end{definition}

\begin{definition}\label{CSS-def}
A set of sequences $S=\{\bm{f}_0, \bm{f}_1,\cdots, \bm{f}_{N-1}\}$ is called  a {\em complementary sequence set} (CSS) of size $N$ if
\begin{equation}
\sum_{u=0}^{N-1}C_{\bm{f}_u}(\tau)=0\  \mbox{for}\  \tau\neq 0.
\end{equation}
\end{definition}
If the set size $N=2$, such a set is called a {\em Golay sequence pair} (GSP). Each sequence in GSP is called a {\em Golay sequence}.

Two CSSs $S_1=\{\bm{f}_{1,0}, \bm{f}_{1,1},\cdots, \bm{f}_{1,N-1}\}$ and $S_2=\{\bm{f}_{2,0}, \bm{f}_{2,1},\cdots, \bm{f}_{2,N-1}\}$ are said to be {\em mutually orthogonal}
if
\begin{equation}\label{CCC-def}
\sum_{v=0}^{N-1}C_{\bm{f}_{1,v},\bm{f}_{2,v}}(\tau)=0 \ \mbox{for}\ \forall \tau.
\end{equation}
It is known that the number of CSSs which are
pairwise mutually orthogonal is at most equal to $N$, the number of sequences
in a CSS.

\begin{definition}
Let $S_u=\{\bm{f}_{u,0}, \bm{f}_{u,1},\cdots, \bm{f}_{u,N-1}\}$  be CSSs of size $N$ for $0\leq u<N$, which are pairwise mutually orthogonal. Such a collection of $S_u$ is called {\em complete mutually orthogonal complementary sets} (CMOCS) or {\em complete complementary codes} (CCC).
\end{definition}

The concept of CCC is better to view through a matrix whose $u$th row is given by $S_u$, i.e.,
\begin{equation} \label{matrix-seq}
\begin{bmatrix}
  \bm{f}_{0,0} & \bm{f}_{0,1} & \cdots & \bm{f}_{0,N-1} \\
  \bm{f}_{1,0} & \bm{f}_{1,1} & \cdots & \bm{f}_{1,N-1} \\
  \vdots & \vdots & \ddots & \vdots \\
  \bm{f}_{N-1,0} & \bm{f}_{N-1,1} & \cdots & \bm{f}_{N-1,N-1} \\
\end{bmatrix}.
\end{equation}

\subsection{CCA, CAS and PU Matrix}

The concept of GSP was generalized to {\em Golay array pair} (GAP) in \cite{Luke}. Moreover, a powerful three-stage process was presented in \cite{Array2}
by showing that all known standard \cite{DavisJedwab99} and non-standard \cite{Fiedler06,Fiedler2,Li05} Golay sequences  of length $2^m$   can be derived from the seed GAPs. Inspired by the excellent idea  for array pair  \cite{Array2, Parker11}, the concepts of CSS and CCC were generalized from sequence to array in \cite{CCA}. Here we introduce them from the  viewpoint of the generating function.

An $m$-dimensional $q$-ary array of size $\underbrace{p\times p \times \cdots \times p}_m$ can be represented by a {\em corresponding function} from $\Z_p^m$ to $\Z_q$:
$$f(\bm{x})=f(x_0, x_1, \cdots, x_{m-1}): \Z_p^m \rightarrow \Z_q,$$
where $\bm{x}=(x_0, x_1, \cdots, x_{m-1})$ and $x_k\in \Z_{p}$.
Let $\bm{z}=(z_0, z_1, \cdots, z_{m-1})$ and  $\bm{z}^{-1}=(z_0^{-1},z_1^{-1},\cdots, z_{m-1}^{-1})$.  The generating function of  an array $f(\bm{x})$ is defined by
\begin{equation}\label{array-gene}
F(\bm{z})=\sum_{x_0=0}^{p-1}\sum_{x_1=0}^{p-1}\cdots \sum_{x_{m-1}=0}^{p-1}\omega^{f(\bm{x})}z_0^{x_0}z_1^{x_1}\cdots z_{m-1}^{x_{m-1}}.
\end{equation}

\begin{definition}\label{CAS-def}
A set of arrays $\{f_0(\bm{x}), f_1(\bm{x}),\cdots, f_{N-1}(\bm{x})\}$ from $\Z_p^m$ to $\Z_q$ is called  a {\em complementary array set} (CAS) of size $N$ if their generating functions $\{F_0(\bm{z}), F_1(\bm{z}),\cdots, F_{N-1}(\bm{z})\}$ satisfy
\begin{equation}
\sum_{u=0}^{N-1}F_u(\bm{z})\cdot \overline{F}_u(\bm{z}^{-1})=N\cdot p^{m}.
\end{equation}
\end{definition}

Two CASs $S_1=\{f_{1,0}(\bm{x}), f_{1,1}(\bm{x}),\cdots, f_{1,N-1}(\bm{x})\}$ and $S_2=\{f_{2,0}(\bm{x}), f_{2,1}(\bm{x}),\cdots, f_{2,N-1}(\bm{x})\}$ are said to be {\em mutually orthogonal}
if  their generating functions  $\{F_{u,0}(\bm{z}), F_{u,1}(\bm{z}),\cdots, F_{u, N-1}(\bm{z})\}$ ($u=1,2$) satisfy
\begin{equation}
\sum_{v=0}^{N-1} F_{1,v}(\bm{z})\overline{F}_{2,v}(\bm{z}^{-1})=0.
\end{equation}

\begin{definition}\label{CCA-def}
Let $S_u=\{f_{u,0}(\bm{x}), f_{u,1}(\bm{x}),\cdots, f_{u,N-1}(\bm{x})\}$ ($0\leq u<N$) be CASs  of size $N$, which  are pairwise mutually orthogonal. We call such a collection of $S_u$ ($0\leq u<N$) a {\em complete mutually orthogonal array set} or a {\em complete complementary arrays} (CCA).
\end{definition}

Let $\widetilde{\bm{M}}(\bm{x})$ be a matrix where  the $u$th row is given by $S_u$, i.e.,
\begin{equation} \label{matrix-array}
\widetilde{\bm{M}}(\bm{x})=\begin{bmatrix}
  f_{0,0}(\bm{x}) & f_{0,1}(\bm{x}) & \cdots & f_{0,N-1}(\bm{x}) \\
  f_{1,0}(\bm{x}) & f_{1,1}(\bm{x}) & \cdots & f_{1,N-1}(\bm{x}) \\
  \vdots & \vdots & \ddots & \vdots \\
  f_{N-1,0}(\bm{x}) & f_{N-1,1}(\bm{x}) & \cdots & f_{N-1,N-1}(\bm{x}) \\
\end{bmatrix}.
\end{equation}
The generating functions of the entries in matrix $\widetilde{\bm{M}}({\bm{x}})$ can be  presented by  a matrix $\bm{M}({\bm{z}})$ with each entry given by ${{M}}_{u,v}(\bm{z})={F}_{u,v}(\bm{z})$, the generating function of ${f}_{u,v}(\bm{x})$, i.e.,
\begin{equation}\label{gene-matrix}
\bm{M}({\bm{z}})=	\begin{bmatrix}
F_{0,0}({\bm{z}})&F_{0,1}({\bm{z}})&\dots&F_{0,N-1}({\bm{z}})\\
F_{1,0}({\bm{z}})&F_{1,1}({\bm{z}})&\dots&F_{1,N-1}({\bm{z}})\\
\vdots             &\vdots             &\ddots&\vdots             \\
F_{N-1,0}({\bm{z}})&F_{N-1,1}({\bm{z}})&\dots&F_{N-1,N-1}({\bm{z}})\\
\end{bmatrix}.
\end{equation}
$\bm{M}({\bm{z}})$ is called the generating  matrix of $\widetilde{\bm{M}}({\bm{x}})$. If $\widetilde{\bm{M}}({\bm{x}})$ is a CCA, it is necessary that its generating matrix
$\bm{M}(\bm{z})$ is a  {\em multivariate para-unitary  (PU) matrix}, i.e., $\bm{M}(\bm{z})\cdot\bm{M}^{\dagger}(\bm{z}^{-1})=c\cdot \bm{I}_N$, where $(\cdot)^\dagger$ denotes the Hermitian transpose, $\bm{I}_N$ is the identity matrix of order $N$, and $c$ is a real constant.

For an array (function) $f(\bm{x})=f(x_0, x_1, \cdots, x_{m-1}): \Z_{p}^m \rightarrow \Z_q$, by ordering of the elements $\bm{x}$ in $\Z_{p}^m$, $f(\bm{x})$ can be  associated  with a sequence $\bm{f}$ of length $L=p^m$, where
$$f(t=\sum_{k=0}^{m-1}x_k\cdot p^k)=f(\bm{x}).$$
In this paper,  we say that  the sequence $f(t)$ is evaluated by the function $f(\bm{x})$. CSSs and CCCs can be constructed form a single CCA.
\begin{fact}\label{array-to-sequence} (\cite{CCA})
For a given CCA $\widetilde{\bm{M}}(\bm{x})$ in the matrix form (\ref{matrix-array}) and an arbitrary permutation $\pi$ of symbol $\{0, 1, \cdots, m-1\}$, let $\widetilde{\bm{M}}(\pi\cdot\bm{x})$ be a matrix with entry $\widetilde{\bm{M}}_{u, v}(\pi\cdot\bm{x})=f_{u,v}(\pi\cdot \bm{x})$, where $\pi\cdot\bm{x}=(x_{\pi(0)},x_{\pi(1)},\cdots, x_{\pi(m-1)})$. Then we have
\begin{itemize}
\item[(1)] Sequences evaluated by functions in $\widetilde{\bm{M}}(\pi\cdot\bm{x})$ form a CCC.
\item[(2)] Sequences evaluated by functions in each row (or column) of $\widetilde{\bm{M}}(\pi\cdot\bm{x})$ form a CSS.
\end{itemize}
\end{fact}

\subsection{Algebraic Normal Form and Generalized Reed Muller Codes}

We are particularly interested in the case that the function $f(\bm{x})$ from $\Z_p^m$  to $\Z_q$ can be realized by {\em algebraic normal form} (ANF).

For $p=2$, the function $f(\bm{x})$ can be realized by {\em generalized Boolean function} (GBF) from $\F_{2}^m$ to $\Z_q$. Every such function can be written in ANF as a sum of monomials of the form $x_{0}^{i_0}x_{1}^{i_1}\cdots x_{m}^{i_m}$ over $\Z_q$, where $i_k=0$ or $1$ for $0\leq k\leq m-1$.

For $q=p$ prime, the function $f(\bm{x})$ can be realized by function from $\F_{p}^m$ to $\F_p$. Every such function can be written in ANF as a sum of monomials of the form $x_{0}^{i_0}x_{1}^{i_1}\cdots x_{m}^{i_m}$ over $\F_p$, where $0\leq i_k\leq p-1$ for $0\leq k\leq m-1$.

The $r$th-order classical Reed-Muller code RM($r, m$) \cite{MacWilliams} is defined to be the binary code whose codewords are (the sequences evaluated by) the Boolean functions of degree at most $r$ in variables $x_0, x_1, \cdots, x_{m-1}$. The code RM($r, m$) is linear, has minimum Hamming distance $2^{m-r}$.
There are two types of  nonbinary generalizations of the classical Reed-Muller codes.

The first type of generalization \cite{DavisJedwab99, Paterson00}, denoted by RM$_q$($r, m$), directly generalize alphabets  of classical Reed-Muller codes from binary to $q$-ary  case. RM$_q$($r, m$) is defined to be a linear code over $\Z_q$ comprised of all the $q$-ary sequences evaluated by ANF of GBFs from $\F_{2}^m$ to $\Z_q$ with degree less than or equal to $r$. It is known that RM$_q$($r, m$) also has minimum Hamming distance $2^{m-r}$.

The another type of generalization, denoted by GRM$_p$($r, m$) in this paper, was studied in \cite{Kasami68, Delsarte1970}. GRM$_p$($r, m$) is defined to be a linear code over $\F_p$ comprised of all the $p$-ary sequences evaluated by ANF of functions from $\F_{p}^m$ to $\F_p$ with degree less than or equal to $r$. It is known in \cite{Kasami68} that GRM$_p$($r, m$) has minimum Hamming distance $(R+1)\cdot p^{Q}$, where $R$ is the remainder and $Q$ the quotient resulting from dividing $m(p-1)-r$ by $p-1$.

\subsection{Hadamard Matrices and  Generalized Seed PU Matrices}

A complex  matrix $\bm{H}$ of order $N$ is
called {\em Butson-type}  Hadamard (BH) matrix \cite{Butson62} if $\bm{H}\cdot\bm{H}^{\dagger}=N\cdot \bm{I}_N$  and all the
entries of \(\bm{H}\) are $q$th roots of unity. For given $N$ and $q$,  the set of all BH matrices is denoted by $H(q,N)$.

Two BH matrices, $\bm{H}_1, \bm{H}_2\in H(q,N)$ are
called {\em equivalent}, denoted by $\bm{H}_1\simeq \bm{H}_2$, if there exist
diagonal unitary matrices $\bm{Q}_1, \bm{Q}_2$ where each diagonal entry
 is a $q$th root of unity and permutation matrices \(\bm{P}_1\),
\(\bm{P}_2\) such that $\bm{H}_1 = \bm{P}_1 \cdot\bm{Q}_1 \cdot\bm{H}_2\cdot \bm{Q}_2 \cdot \bm{P}_2$.

For a BH matrix $\bm{H}\in H(q,N)$, define its phase matrix $\widetilde{\bm{H}}$ by
$\widetilde{H}_{i,j}=s$ if $H_{i,j}=\omega^s$. Suppose that $S_{\widetilde{H}}(q, N)$ is a set containing all the phase matrix of the representatives of BH matrices in $H(q,N)$ with respected to the equivalence relation.

In this and the next subsections, we revisit and extend a method  proposed in \cite{CCA} to construct the so  called {\em generalized seed PU matrices}. Note that the order of the seed PU matrices and the generalized seed PU matrices in \cite{CCA} are $N=p$ and $N=2^n$, respectively, while the order of the generalized seed PU matrices here is straightforwardly extended to $N=p^n$ for arbitrary $p$.

The {\em delay matrix} $\bm{D}(z)$  of order $p$ is defined by  $\bm{D}(z)=diag\{z^0, z^1, z^2,$ $\cdots, z^{p-1}\}$. And the generalized delay matrix $\bm{D}(\bm{z})$ with multi-variables $\bm{z}=(z_0, z_{1}, \cdots,  z_{n-1})$
is defined by the Kronecker product of $\bm{D}(z_j)$ for $0\leq j<n$, i.e.,
\begin{equation}\label{delay}
\bm{D}(\bm{z})=\bm{D}(z_{n-1}) \otimes\cdots \otimes \bm{D}(z_{1})\otimes \bm{D}(z_{0}).
\end{equation}
By mathematical induction, it is straightforward to show the generalized delay matrix $\bm{D}(\bm{z})$ can be explicitly  expressed by a diagonal matrix $\bm{D}(\bm{z})=diag\{\phi_0(\bm{z}),\phi_1(\bm{z}), \cdots, \phi_{p^n-1}(\bm{z})\}$
where $\phi_y(\bm{z})=\prod_{j=0}^{n-1}z_{j}^{x_j}$ for $y=\sum_{j=0}^{n-1}x_j\cdot p^j$.

Let $\bm{H}^{\{k\}}$ be arbitrary BH matrices chosen from $H(q, N)$ for $0\leq k\leq m$ and $\bm{D}(\bm{z}_k)$ be the generalized delay matrices with $\bm{z}_k=(z_{kn}, z_{kn+1}, \cdots,  z_{kn+n-1})$ for $0\leq k<m$. Let $\bm{z}=(z_{0}, z_{1}, \cdots,  z_{nm-1})$.
It has been shown in \cite{CCA} that a multivariate polynomial matrix $\bm{M}(\bm{z})$, defined by
\begin{equation}\label{seed-PU-2}
\bm{M}(\bm{z})=\bm{H}^{\{0\}}\cdot \bm{D}(\bm{z}_0)\cdot
\bm{H}^{\{1\}}\cdot\bm{D}(\bm{z}_1)\cdots \bm{H}^{\{m-1\}}\cdot
\bm{D}(\bm{z}_{m-1})\cdot \bm{H}^{\{m\}},
\end{equation}
 must be the generating matrix of a CCA, denoted by $\widetilde{\bm{M}}(\bm{x})$ where $\bm{x}=(x_0,x_1,\cdots, x_{mn-1})$. $\bm{M}(\bm{z})$ is called the generalized seed PU matrix.

\subsection{Extracting Functions from Generalized Seed PU Matrices}

Let $\bm{x}_k= (x_{kn}, x_{kn+1}, \cdots, x_{kn+n-1})\in \Z_p^n$ be the $p$-ary expansion of $y_k$ and $\bm{y}=(y_0, y_{1}, \cdots,  y_{m-1})$. Each entry of $\widetilde{\bm{M}}(\bm{x})$ extracted from the generalized seed PU matrix $\bm{M}(\bm{z})$ can be represented by a function $f_{u,v}(\bm{y})$  from $\Z_{p^n}^m$ to $\Z_q$ determined by the formula
$$\omega^{f_{u,v}(\bm{y})}={H}^{\{0\}}_{u,y_0}\cdot\left(\prod^{m-1}_{k=1}{H}^{\{k\}}_{y_{k-1},y_k}
\right)\cdot {H}^{\{m\}}_{y_{m-1},v}.$$
The function
 $f_{u,v}(\bm{y})$  can be alternatively presented by $f_{u,v}(\bm{x})$, which is an array from $\Z_{p}^{mn}$ to $\Z_q$. The general form of the entries of $\widetilde{\bm{M}}(\bm{x})$, denoted by $f(\bm{x})$ (or $f(\bm{y})$),  is given  by introducing a basis of the functions from $\Z_{p^n}$ to  $\mathbb{Z}_q$.

\begin{definition}\label{basis}
Let $\delta_{\alpha}(y)$ be a function: $\Z_{p^n}\rightarrow \mathbb{Z}_q$ such that $\delta_{\alpha}(\beta)=\delta_{\alpha,\beta}$ where $\delta_{\alpha,\beta}$ is the Kronecker-delta function, i.e.,
$$\delta_{\alpha}(\beta)=\left\{
\begin{aligned}
&1, \mbox{if} \ \alpha=\beta, \\
&0, \mbox{if} \ \alpha\neq \beta.
\end{aligned} \right.$$
\end{definition}

The function $f(\bm{y})$  can be explicitly  represented by the combination of the  linear terms and the   quadratic terms with respect to the functions $\delta_{\alpha}(y_{k})$ for $\alpha\in \Z_{p^n}$ and $0\leq k\leq m-1$.

\begin{definition}\label{linear-term-def}
The {\em $\delta$-linear terms} are  the linear combinations of $\delta_{\alpha}(y_{k})$ over $\Z_q$ for $\alpha\in \Z_{p^n}$ and $0\leq k\leq m-1$. The collection of the $\delta$-linear terms is denoted by
\begin{equation}\label{linear-term-1}
\delta_L(q, p^n)=\left\{\sum_{k=0}^{m-1}\sum_{\alpha=0}^{p^n-1}c_{\alpha, k}\cdot \delta_{\alpha}(y_{k})\bigg| \forall c_{\alpha,k}\in \Z_q\right\}.
\end{equation}
\end{definition}
It has been shown in \cite[Lemma 7]{CCA} that the $\delta$-linear terms can be represented by
\begin{equation}\label{linear-term-2}
\delta_L(q, p^n)=\left\{\sum_{k=0}^{m-1}\sum_{\alpha=1}^{p^n-1}c_{\alpha, k}\cdot \delta_{\alpha}(y_{k})+c'\bigg| c_{\alpha,k}, c'\in \Z_q \right\},
\end{equation}
which is a free $\Z_q$-submodule of dimension $m(p^n-1)+1$ with basis $\{\delta_{\alpha}(y_{k}), 1| 0\leq k\leq m-1, 1\leq \alpha\leq q^n-1 \}$.

Let $\chi$ be a permutation of symbols $\{0, 1, \cdots, p^n-1\}$ and $\bm{\delta}(y)=(\delta_{0}(y), \delta_{1}(y), \cdots, \delta_{p^n-1}(y))$.
Denote the permutation of the vector function $\bm{\delta}(y)$ by
\begin{equation}
\bm{\delta}_{\chi}(y)=(\delta_{\chi(0)}(y), \delta_{\chi(1)}(y), \cdots, \delta_{\chi(p^n-1)}(y)).
\end{equation}
 The   {\em $\delta$-quadratic terms} can be obtained  from the phase matrices $\widetilde{\bm{H}}\in S_{\widetilde{H}}(q, p^n)$.

\begin{definition}\label{quadratic-term-def}
The   {\em $\delta$-quadratic terms} are quadratic forms:
\begin{equation}\label{quadratic-term}
\bm{g}_{\chi_{L}}(y_{0})\widetilde{\bm{H}}\bm{g}_{\chi_{R}}(y_{1})^T,
\end{equation}
where $\chi_L, \chi_R$ are permutations of symbols $\{0, 1, \cdots, p^n-1\}$  and $\widetilde{\bm{H}}\in S_{\widetilde{H}}(q, p^n)$. The collection of the $\delta$-quadratic terms is denoted by $\delta_Q(q, p^n)$.
\end{definition}

\begin{fact}\label{fact-2}
All the  functions extracted from the generalized seed PU matrices  can be represented in a general form
\begin{equation}\label{func-rep}
f(\bm{x})=\sum_{k=1}^{m-1}h_k(y_{k-1}, y_{k})+l(\bm{y}),
\end{equation}
where $h_k(\cdot,\cdot)\in \delta_Q(q, p^n) (1\leq k\leq m-1)$ and $l(\bm{y})\in \delta_L(q, p^n)$.
\end{fact}

\begin{fact}\label{fact-3}
Let $f(\bm{x})$ be a function (or array of size $\underbrace{p\times p \times \cdots \times p}_{mn}$) with the form (\ref{func-rep}) and $h_0(\cdot,\cdot), h_m(\cdot,\cdot)\in \delta_Q(q, p^n)$. Then the arrays
$$f_u(\bm{x})=f(\bm{x})+h_0(u,y_0), \ u\in \Z_{p^n}$$
form a CAS of size $p^n$, and the arrays
$$f_{u,v}(\bm{x})=f(\bm{x})+h_0(u,y_0)+h_m(y_m,v), \ u, v\in \Z_{p^n}$$
form a CCA.
\end{fact}

\section{$\delta$-Linear Terms and $\delta$-Quadratic Terms}

CSSs and CCCs are constructed by Facts 1-3, the kernel of which are the  $\delta$-linear terms and $\delta$-quadratic terms. However, it have been shown in \cite{CCA} that these terms are calculated by heavy computation even for $p=n=2$. In this section, we will continue our study on $\delta$-linear terms and $\delta$-quadratic terms to avoid the computation of the Kronecker-delta functions in Definition \ref{basis}.

\subsection{$\delta$-Linear Terms}

From the definition of the function $\delta_{\alpha}(\cdot)$, any function $g(y)$: $\Z_{p^n}\rightarrow \mathbb{Z}_q$ can be represented by
$$g(y)=\sum_{\alpha\in \Z_{p^n}}g(\alpha)\delta_{\alpha}(y).$$
Then it is clear that the function $\sum_{\alpha\in \Z_{p^n}}c_{\alpha, k}\cdot \delta_{\alpha}(y_{k})$ in (\ref{linear-term-1}) can be expressed by a function  $l_k(y_{k}):\Z_{p^n}\rightarrow \mathbb{Z}_q$ such that $l_k(\alpha)=c_{\alpha, k}$.
\begin{theorem}\label{thm-l}
The collection of the $\delta$-linear terms can be represented in an alternative form:
\begin{equation*}
\delta_L(q, p^n)=\left\{\sum_{k=0}^{m-1}l_k(y_{k}) \bigg| \forall l_k(y_{k}):\Z_{p^n}\rightarrow \mathbb{Z}_q \right\}.
\end{equation*}
\end{theorem}

We are particularly interested in the case that the functions  in $\delta_L(q, p^n)$ can be realized by ANF introduced in Subsection 2.3. Recall that $\bm{x}_k= (x_{kn}, x_{kn+1}, \cdots, x_{kn+n-1})\in \Z_p^n$ be the $p$-ary expansion of $y_k\in \Z_{p^n}$.

\begin{corollary}\label{coro-1}
For the case $p=2$, the variables $y_k$ over $\Z_{2^n}$ can be replaced by $\bm{x}_k\in \F_2^n$, and
any function $l_k(y_k)$ from $\Z_{2^n}$ to $\Z_q$ can be represented by a GBF from $\F_{2}^{n}$ to $\Z_q$. Then we have
 \begin{equation*}
\delta_L(q, p^n=2^n)=\left\{ \sum_{k=0}^{m-1}\sum_{i=1}^{2^n-1}\left(c_{k,i}\cdot\prod_{j=0}^{n-1}x_{kn+j}^{i_j}\right)+c' \bigg| c_{k,i}, c'\in \Z_q\right\},
\end{equation*}
where  $(i_0,i_1,\cdots,i_{n-1})$ is the binary expansion of integer $i$.
\end{corollary}

\begin{corollary}\label{coro-2}
For the case $q=p$ prime, the variables $y_k$ over $\Z_{p^n}$ can be replaced by $\bm{x}_k\in \F_p^n$, and
any function $l_k(y_k)$ from $\Z_{p^n}$ to $\F_p$ can be represented by an ANF from $\F_{p}^{n}$ to $\F_p$. Then we have
 \begin{equation*}
\delta_L(q=p, p^n)=\left\{ \sum_{k=0}^{m-1}\sum_{i=1}^{p^n-1}\left(c_{k,i}\cdot\prod_{j=0}^{n-1}x_{kn+j}^{i_j}\right)+c' \bigg| c_{k,i}, c'\in \F_p\right\},
\end{equation*}
where  $(i_0,i_1,\cdots,i_{n-1})$ is the $p$-ary expansion of integer $i$.
\end{corollary}

\begin{remark}
Both Corollaries 1 and 2 change the basis from the Kronecker-delta functions shown in Definition \ref{basis} to classical basis of monomials shown in Subsection 2.3. Moreover, from the basis of monomials, it is obvious that $\delta_L(q, p^n=2^n)$ and  $\delta_L(q=p, p^n)$ are  $\Z_q$-submodule of dimension $m(2^n-1)+1$  and $m(p^n-1)+1$, respectively, which agree with the results in \cite[Lemma 7]{CCA}.
\end{remark}

\subsection{$\delta$-Quadratic Terms}

Recall that  the   $\delta$-quadratic terms are quadratic forms:
$\bm{g}_{\chi_{L}}(y_{0})\widetilde{\bm{H}}\bm{g}_{\chi_{R}}(y_{1})^T$.  We will show in the rest of paper that, if we take the algebraic structure of BH matrices into consideration, the computation of the ANF of $\delta$-quadratic terms can be significantly simplified.

\begin{theorem}\label{theorem-q}
Let $\bm{H}$ be a BH matrix with entry $H_{u, v}=\omega^{h(u, v)}$ where $h(u, v)$ is a function from $\Z_{p^n}^2$ to $\Z_q$ for $u,v\in \Z_{p^n}$, and  $g(\cdot), g'(\cdot)$ be arbitrary  permutation functions over $\Z_{p^n}$. We have
\begin{equation}\label{qudratic-term-0}
h(g(y_0), g'(y_1))\in \delta_Q(q, p^n).
\end{equation}
\end{theorem}
\begin{proof}
We take $\bm{H}$ as a representative of its equivalence class of BH matrices. Then the entry of its phase matrix $\widetilde{\bm{H}}$ is given by $\widetilde{H}_{u, v}=h(u, v)$ for $u,v\in \Z_{p^n}$.  We have
$$\bm{g}_{\chi_{L}}(y_{0})\widetilde{\bm{H}}\bm{g}_{\chi_{R}}(y_{1})^T=\sum_{u\in \Z_{p^n}}\sum_{v\in \Z_{p^n}}h(u, v)\delta_{\chi_L(u)}(y_0)\delta_{\chi_R(v)}(y_1),$$
where $\chi_L, \chi_R$ are permutations of symbols $\{0, 1, \cdots, p^n-1\}$.
We define the Kronecker-delta function  $\delta_{\alpha, \beta}(y_0, y_1)$ from $\Z_{p^n}^2$ to $\mathbb{Z}_q$ such that $\delta_{\alpha, \beta}(y_0, y_1)=\delta_{\alpha}(y_0)\cdot \delta_{\beta}(y_1)$, i.e.,
$$\delta_{\alpha, \beta}(y_0, y_1)=\left\{
\begin{aligned}
&1, \mbox{if} \ (\alpha, \beta)=(y_0, y_1), \\
&0, \mbox{if} \ (\alpha, \beta)\neq(y_0, y_1).
\end{aligned} \right.$$
Then any function $h(y_0, y_1)$: $\Z_{p^n}^2\rightarrow \mathbb{Z}_q$ can be represented by
$$h(y_0, y_1)=\sum_{\alpha\in \Z_{p^n}}\sum_{\beta\in \Z_{p^n}}h(\alpha, \beta)\cdot\delta_{\alpha, \beta}(y_0, y_1).$$
Thus, we have
\begin{eqnarray*}
\bm{g}_{\chi_{L}}(y_{0})\widetilde{\bm{H}}\bm{g}_{\chi_{R}}(y_{1})^T&=&\sum_{u\in \Z_{p^n}}\sum_{v\in \Z_{p^n}}h(u, v)\cdot\delta_{\chi_L(u), \chi_R(v)}(y_0, y_1)\\
&=&\sum_{u\in \Z_{p^n}}\sum_{v\in \Z_{p^n}}h(\chi_L^{-1}(u), \chi_R^{-1}(v))\cdot\delta_{u, v}(y_0, y_1)\\
&=&h(\chi_L^{-1}(y_0), \chi_R^{-1}(y_1)).
\end{eqnarray*}
Let $\chi_{L}$ and $\chi_{R}$ be the inverse functions of the permutation functions $g(\cdot)$ and $g'(\cdot)$ over $\Z_{p^n}$, respectively. We complete the proof.
\end{proof}

\section{Constructions from  DFT Matrices over $\Z_N$}

In this section, we assume $N=p^n=q$. Discrete Fourier transform (DFT) matrix of order $N$ is a BH matrix with entry $H_{u, v}=w^{u\cdot v}$ for $u, v\in \Z_N$. Then the entry of its  phase matrix is given by $\widetilde{H}_{u, v}=u\cdot v$.  We have
$$g(y_0)g'(y_1)\in \delta_Q(N, N)$$ for arbitrary  permutation functions $g(\cdot), g'(\cdot)$  over $\Z_{p^n}$. These terms are called  the  $\delta$-quadratic terms determined by DFT matrices.

\begin{remark}
From the arguments on the equivalence of $\delta$-quadratic terms in \cite{CCA}, there are totally $\varphi(N)\cdot((N-1)!)^2$ $\delta$-quadratic terms determined by DFT matrices, where $\varphi(N)$ is the Euler function of integer $N$. These functions will be explicitly given in this section for $N$ prime or a power of $2$.
\end{remark}

Recall that the set of the $\delta$-linear terms in Theorem 1 is given by
\begin{equation}
\delta_L(q=N, p^n=N)=\left\{\sum_{k=0}^{m-1}l_k(y_{k}) \bigg| \forall l_k(y_{k}):\Z_{N}\rightarrow \mathbb{Z}_N \right\}.
\end{equation}
We obtain the following results immediately by Facts 2 and 3.

\begin{theorem}
Let $g_k(\cdot), g_k'(\cdot)$ are arbitrary  permutation functions over $\Z_{N}$ for  $0\leq k\leq m$, $l(\bm{y})$ arbitrary functions in the set $\delta_L(N, N)$, and $f(\bm{y})$  an $N$-ary function  with the form
\begin{equation}
f(\bm{y})=\sum_{k=1}^{m-1}g_{k}(y_{k-1})\cdot g'_{k}(y_{k})+l(\bm{y}).
\end{equation}
  \begin{itemize}
	\item[(1)] The following functions from $\Z_N$ to $\Z_N$ form a CAS of size $N$:
\begin{equation}
f_u(\bm{y})=f(\bm{y})+u\cdot g_0(y_0) \ \mbox{for}\  u\in \Z_N.
\end{equation}
 \item[(2)] The following functions from $\Z_N$ to $\Z_N$ form a CCA of size $N$:
\begin{equation}
f_u(\bm{y})=f(\bm{y})+u\cdot g_0(y_0)+v\cdot g_{0}'(y_{m-1}) \ \mbox{for}\  u\in \Z_N.
\end{equation}
\end{itemize}
\end{theorem}

In the rest of the section, we study the case that the functions $f(\bm{y})$ can be realized by ANF shown in Subsection 2.3.

\subsection{Constructions from PPs over $\F_p$}

In this subsection, we will set $N=q=p$ prime and $n=1$. Then we have $g_k(\cdot), g_k'(\cdot)$ are arbitrary PPs over $\F_{p}$,  $y_k=x_k$ and $\bm{x}=(x_0,x_1,\cdots, x_{m-1})$.
From Corollary 2, the set of the $\delta$-linear terms  is given by
 \begin{equation}\label{linear-1}
\delta_L(q=p, p^n=p)=\left\{ \sum_{k=0}^{m-1}\sum_{i=1}^{p-1}c_{k, i}x_k^i+c'\bigg| c_{k,i}, c'\in \F_p\right\}.
\end{equation}

\begin{theorem}\label{thm-4}
Let $g_k(\cdot), g_k'(\cdot)$ are arbitrary  PPs over $\F_{p}$ for  $0\leq k\leq m$, $\pi$  arbitrary  permutation of symbols $\{0, 1, \cdots m-1\}$, and $\forall l(\bm{x})\in \delta_L(p, p)$. For any $p$-ary function $f(\bm{x})$ with the form
\begin{equation}\label{array-1}
f(\bm{x})=\sum_{k=1}^{m-1}g_{k}(x_{k-1})\cdot g'_{k}(x_{k})+l(\bm{x}),
\end{equation}
  we have the following results.
\begin{itemize}
	\item[(1)] The $p$-ary sequences  evaluated by  functions from $\F_p^m$ to $\F_p$:
\begin{equation}\label{CSS-1}
f_u(\bm{x})=f(\pi\cdot\bm{x})+u\cdot g_0(x_{\pi(0)}) \ \mbox{for}\  u\in \F_p,
\end{equation}
form a CSS of size $p$.
   \item[(2)] The  $p$-ary sequences  evaluated by functions from $\F_p^m$ to $\F_p$:
\begin{equation}\label{CCC-1}
f_{u,v}(\bm{x})=f(\pi\cdot\bm{x})+u\cdot g_0(x_{\pi(0)})+v\cdot g_{0}'(x_{\pi(m-1)}) \ \mbox{for}\  u, v\in \F_p,
\end{equation}
form a CCC.
\end{itemize}
\end{theorem}

The above theorem is valid, since  $g_0(x_0)\cdot g_1(x_1)\in \delta_Q(p, p)$ is a $\delta$-quadratic term if $g_0(\cdot)$ and $g_1(\cdot)$ are PPs over $\F_{p}$. However, different $\delta$-quadratic terms may result in the same sequence if their difference is a $\delta$-linear term \cite{CCA}. To avoid duplication of the sequences in Theorem \ref{thm-4}, we introduce the following PPs over finite field $\F_{p^n}$ for $p$ prime.

\begin{definition}\label{semi}
A polynomial $g(x)$ over $\F_{p^n}$ is called a {\em semi-normalized} PP if $g(x)$ is a monic PP and  $g(0)=0$. The collection of all  semi-normalized PPs over $\F_{p^n}$ are denoted by $\mathcal{S}^{(p^n)}$ in this paper.
\end{definition}

It is obvious that the number of the semi-normalized PPs  is $(p^n-2)!$. If $g(x)$ is a PP and $a\neq 0, b\neq 0, c\in\F_{p^n}$, then $g_1(x)=a\cdot g(bx+c)$ is also a PP. By suitably choosing $a, b, c$,  we can arrange to have $g_1(x)$ in {\em normalized form} so that $g_1(x)$ is monic, $g_1(0)=0$, and when the degree $d$ of $g_1(x)$ is not divisible by $p$, the coefficient of $x^{d-1}$ is 0. Normalized PPs are well-studied in the literature. For example, a list of all normalized PPs of degree at most $5$ can be found in \cite[Ch.8]{Mullen}, and all normalized PPs of degree 6 was tabulated in \cite{Winterhof}. For any semi-normalized PP $g(x)$, it is obvious that there exists normalized PP $g_1(x)$ and $c, e\in\F_{p^n}$ such that $g(x)=g_1(x+c)+e$, so all the semi-normalized PPs in Definition \ref{semi} can be obtained by the well-studied {\em normalized }PPs  over $\F_{p^n}$.

\begin{example}
For $p=5$ and $n=1$, there are two  normalized PPs from  \cite[Ch.8]{Mullen}: $x, x^3$. Then  $3!=6$ semi-normalized  PPs in $\mathcal{S}^{(5)}$ can be obtained as follows.
$$
\begin{aligned}
\mathcal{S}^{(5)}=&\{x, x^3, (x+1)^3+4, (x+2)^3+2, (x+3)^3+3, (x+4)^3+1\}\\
=&\{x, x^3, x^3+3x^2+3x, x^3+x^2+2x, x^3+4x^2+2x, x^3+2x^2+3x\}.
\end{aligned}
$$
\end{example}

It is obvious that $dg_0(x_0)g_1(x_1)\in \delta_Q(p, p)$ is a $\delta$-quadratic term for $d\in \F_p^*$ and $g_0(\cdot), g_1(\cdot)\in \mathcal{S}^{(p)}$. Moreover,
it is easy to check that the difference of $dg_0(x_0)g_1(x_1)$ and $d'g_0'(x_0)g_1'(x_1)$ cannot be a $\delta$-linear term if $(d, g_0(\cdot), g_1(\cdot)) \neq (d', g_0'(\cdot), g_1'(\cdot))$.  Then sequences in CSSs determined by DFT over $\F_p$ and the enumeration are given as follows.

\begin{corollary}\label{coro-3}
For $\forall g_k(\cdot), g_k'(\cdot)\in \mathcal{S}^{(p)}$ ($0\leq k\leq m-1$), $d_k\in \F_p^*$, $l(\bm{x})\in \delta_L(p, p)$ and permutation $\pi$, the $p$-ary sequence evaluated by
\begin{equation}\label{seq-1}
f(\bm{x})=\sum_{k=1}^{m-1}d_k\cdot g_{k}(x_{\pi(k-1)})\cdot g'_{k}(x_{\pi(k)})+l(\bm{x})
\end{equation}
lies in a CSS of size $p$. Moreover, formula (\ref{seq-1}) explicitly determines
$$\frac{1}{2}m!\left((p-1)!\right)^{m-1}\left((p-2)!\right)^{m-1}p^{m(p-1)+1}$$
distinct $p$-ary sequences.
\end{corollary}

\begin{proof}
In formula (\ref{seq-1}), there are $(p-1)$ choices of $d_k$, $(p-2)!$ choices of $g_k(\cdot)$ and  $g_k'(\cdot)$ respectively for a fixed $k$, so there are totally $(p-1)!(p-2)!$ $\delta$-quadratic terms determined by DFT matrix of order $p$.
There are $m!$ choices of permutations $\pi$, and $p^{m(p-1)+1}$ choices of function $l(\bm{x})$. Moreover, for two functions with the form (\ref{seq-1}):
\begin{equation*}
f^{(j)}(\bm{x})=\sum_{k=1}^{m-1}d^{(j)}_k\cdot g^{(j)}_{k}(x_{\pi^{(j)}(k-1)})\cdot g'^{(j)}_{k}(x_{\pi^{(j)}(k)})+l^{(j)}(\bm{x})\ \ (j=1,2),
\end{equation*}
we have $f^{(1)}(\bm{x})=f^{(2)}(\bm{x})$ if and only if
\begin{itemize}
	\item[(1)] $\pi^{(1)}=\pi^{(2)}, d^{(1)}_k=d^{(2)}_k, g^{(1)}_k=g^{(2)}_k, g'^{(1)}_k=g'^{(2)}_k, l^{(1)}(\bm{x})=l^{(2)}(\bm{x})$,  or
   \item[(2)] $\pi^{(1)}(k)=\pi^{(2)}(m-k), d^{(1)}_k=d^{(2)}_{m-k}, g^{(1)}_k=g'^{(2)}_{m-k}, g'^{(1)}_k=g^{(2)}_{m-k}, l^{(1)}(\bm{x})=l^{(2)}(\bm{x})$.
\end{itemize}
The proof is completed.
\end{proof}

We firstly give an example for $p=2$ to illustrate the constructions in Theorem \ref{thm-4} and Corollary \ref{coro-3}.
\begin{example}\label{exam-2}
For $p=2$, $\mathcal{S}^{(2)}$ contains only one PP: $g(x)=x$. The functions in Corollary \ref{coro-3} can be expressed by
$$
f(\bm{x})=\sum_{k=1}^{m-1}x_{\pi(k-1)}x_{\pi(k)}+\sum_{k=0}^{m-1}c_{k}x_k+c',
$$
for $c_k, c'\in \F_2$. From Theorem \ref{thm-4}, the sequences evaluated by
$$\left\{
\begin{aligned}
&f(\bm{x}),\\
&f(\bm{x})+x_{\pi(0)}
\end{aligned}\right.
$$
form a binary Golay complementary pair.
\end{example}

The results in Example \ref{exam-2} coincide with the known binary standard Golay sequences, which have been well studied in  \cite{DavisJedwab99}.
In the rest of this subsection, we discuss the case for odd prime $p$. Any $p$-ary function from $\F_p^m$ to $\F_p$ can be treated as a vector  over $\F_p$ of dimension $p^m$. Then $\delta_L(p, p)$ is a subspace (or a code) of dimension $m(p-1)+1$. The collection of the sequences in Corollary \ref{coro-3} is actually the union of the cosets of $\delta_L(p, p)$ with coset leaders:
\begin{equation}\label{leader-1}
\sum_{k=1}^{m-1}d_k\cdot g_{k}(x_{\pi(k-1)})\cdot g'_{k}(x_{\pi(k)}).
\end{equation}
On the other hand, every sequences in Corollary \ref{coro-3}  lies in GRM$_p$($2(p-2), m$) \cite{Kasami68}, since the degree of the PPs over $\F_p$ is no more than $p-2$. The Hamming distance of GRM$_p$($r, m$) has been shown in \cite{Kasami68}. Then a lower bound of the Hamming distance of the union of these cosets is obtained immediately.

\begin{corollary}\label{coro-4}
The Hamming distance of any two distinct sequences in Corollary \ref{coro-3} is at leat $3p^{m-2}$ for odd prime $p$.
\end{corollary}
\begin{proof}
It was shown in \cite[Theorem 5]{Kasami68} that  GRM$_p$($r, m$) has minimum Hamming distance $(R+1)\cdot p^{Q}$, where $R$ is the remainder and $Q$ the quotient resulting from dividing $m(p-1)-r$ by $p-1$. If we set $r=2(p-2)$, we have $Q=m-2$  and $R=2$ for $p>3$, and we have $Q=m-1$ and $R=1$ for $p=3$. Both two cases result in minimum Hamming distance $3p^{m-2}$.
\end{proof}

We continue to give examples for $p=3$ and $5$.
\begin{example}\label{exam-3}
For $p=3$,  $\mathcal{S}^{(3)}$ contains only one PP: $g(x)=x$. Then the functions in Corollary \ref{coro-3} can be expressed by
$$
f(\bm{x})=\sum_{k=1}^{m-1}d_kx_{\pi(k-1)}x_{\pi(k)}+\sum_{k=0}^{m-1}c_{k,2}x_{k}^2+\sum_{k=0}^{m-1}c_{k,1}x_{k}+c',
$$
for  $d_k \in \F_3^*$ and $c_{k,2}, c_{k,1}, c'\in \F_3$. From Theorem \ref{thm-4},  the sequences evaluated by
$$\left\{
\begin{aligned}
& f(\bm{x}),\\
&f(\bm{x})+x_{\pi(0)},\\
&f(\bm{x})+2x_{\pi(0)}
\end{aligned}\right.
$$
form a ternary CSS of size 3.  The sequences evaluated by
$$ f(\bm{x})\cdot \bm{J_3}+
          x_{\pi(0)} \cdot\left(
  \begin{array}{ccccc}
    0 & 0 & 0  \\
    1 & 1 & 1  \\
    2 & 2 & 2  \\
  \end{array}
\right)+x_{\pi(m-1)} \cdot\left(
  \begin{array}{ccccc}
    0 & 1 & 2 \\
    0 & 1 & 2  \\
    0 & 1 & 2  \\
  \end{array}
\right)
$$
form a CCC, where $\bm{J}_N$ is the all $1$ matrix of order $N$. These results coincide with ternary case in \cite[Constuction 2]{CCA}.
\end{example}

\begin{example}\label{exam-4}
For $p=5$, $\mathcal{S}^{(5)}$ has been given in Example 1.
The functions in Corollary \ref{coro-3} can be expressed by
$$
f(\bm{x})=\sum_{k=1}^{m-1}d_k\cdot g_{k}(x_{\pi(k-1)})\cdot g'_{k}(x_{\pi(k)})+\sum_{k=0}^{m-1}\left (c_{k,4}x_{k}^4+c_{k,3}x_{k}^3+c_{k,2}x_{k}^2+c_{k,1}x_{k}\right )+c',
$$
where $g_k(\cdot), g'_k(\cdot)\in \mathcal{S}^{(5)}$, $d_k\in \F_5^*$ and $c_{k,i}, c'\in \F_5$.

For $\forall g_0(\cdot)\in \mathcal{S}^{(5)}$, the sequences evaluated by
$$\left\{
\begin{aligned}
&f(\bm{x}),\\
&f(\bm{x})+g_0(x_{\pi(0)}),\\
&f(\bm{x})+2g_0(x_{\pi(0)}),\\
&f(\bm{x})+3g_0(x_{\pi(0)}),\\
&f(\bm{x})+4g_0(x_{\pi(0)})
\end{aligned}\right.
$$
form a quinary CSS of size 5.

For $\forall g_0(\cdot), g'_0(\cdot)\in \mathcal{S}^{(5)}$,  the sequences evaluated by
$$ f(\bm{x})\cdot \bm{J_5}+g_0(x_{\pi(0)}) \cdot\left(
  \begin{array}{ccccc}
    0 & 0 & 0 & 0 & 0 \\
    1 & 1 & 1 & 1 & 1 \\
    2 & 2 & 2 & 2 & 2 \\
    3 & 3 & 3 & 3 & 3 \\
    4 & 4 & 4 & 4 & 4 \\
  \end{array}
\right)+g'_{0}(x_{\pi(m-1)}) \cdot\left(
  \begin{array}{ccccc}
    0 & 1 & 2 & 3 & 4 \\
    0 & 1 & 2 & 3 & 4 \\
    0 & 1 & 2 & 3 & 4 \\
    0 & 1 & 2 & 3 & 4 \\
    0 & 1 & 2 & 3 & 4 \\
  \end{array}
\right)
$$
form a CCC.
\end{example}

\begin{table*}
\caption{The Comparisons of Code Rate With Konwn Results}
		\centering
		\begin{tabular}{|c|c|c|c|c|}
				\hline
				Codes
				&PMEPR at most
				&$\#$ of subcarriers
				& info. rate
				&code rate \\
				\hline\hline
				Binary \cite{Paterson00}
				&4
				&128
				&0.180
				&0.180 \\
                \hline
                Quaternary \cite{Paterson00}
				&4
				&128
				&0.296
				&0.148 \\
                \hline
                Octary \cite{Paterson00}
				&4
				&128
				&0.405
				&0.135 \\
                \hline
                Quinary in this paper
				&5
				&125
				&0.369
				&0.159 \\
                \hline
                Binary \cite{Paterson00}
				&8
				&128
				&0.172
				&0.172 \\
                \hline
				
			\end{tabular}
	\end{table*}

A table of parameters including PMEPR bound, code rate, information rate for quinary sequences of length 125  given in Example \ref{exam-4}, and for binary, quaternary and octary sequences of length $128$ proposed in \cite{Paterson00} is shown in Table 1. Note that it is difficult to do a normalized comparison with known results, since the length and the bound of the PMEPR of the previous  results are alway power of 2.

\subsection{Constructions from Bijective GBFs}

In this subsection, we will set $p=2$ and $N=2^n=q$. In this case, the explicit form of the $\delta$-linear terms $\delta_L(q=2^n, p^n=2^n)$ have been  shown in Corollary \ref{coro-1}. For $\delta$-quadratic  terms,  any  permutation function $g(y)$ over $\Z_{2^n}$ can be realized a bijective GBF $g(\bm{x})$ from $\F_2^n$ to $\Z_{2^n}$. Then the ANF of the functions in CCA and CAS can be given by Theorem 3.

With the same arguments in the previous subsection to avoid the duplication, we define a subset $\mathcal{S}$ of the bijective GBFs from $\F_2^n$ to $\Z_{2^n}$ such that
\begin{itemize}
	\item[(1)] $g(\bm{x}=\bm{0})=0$ for $\forall g(\bm{x})\in \mathcal{S}$,
   \item[(2)] $g_1(\bm{x})\neq d\cdot g_2(\bm{x})$ for $\forall g_1(\bm{x}), g_2(\bm{x})\in \mathcal{S}$, and $(d, 2)=1$.
\end{itemize}
Then we have $dg(\bm{x_0})g'(\bm{x_1})\in\delta_Q(2^n, 2^n)$ for $g(\bm{x_0}), g'(\bm{x_1})\in \mathcal{S}$, $d\in\Z_{2^n}$ and $d$ odd. There are totally
$\frac{(2^n-1)!}{2^{n-1}}$ bijective GBFs in the set $\mathcal{S}$, which leads to $\frac{((2^n-1)!)^2}{2^{n-1}}$
$\delta$-quadratic terms determined by DFT matrices of order $2^n$. We give an example for $n=2$ to illustrate it.

\begin{example}\label{exam-5}
There are 3 bijective GBFs from $\F_2^2$ to $\Z_{4}$ in the set $\mathcal{S}$, which can be explicitly given by
$$\left\{
\begin{aligned}
&g_1(x_0, x_1)=x_0+2x_1,\\
&g_2(x_0, x_1)=2x_0+x_1,\\
&g_3(x_0, x_1)=2x_0x_1+x_0+3x_1.\\
\end{aligned}\right.
$$
Then there are 18 $\delta$-quadratic terms determined by DFT matrices of order $4$, which can be presented by
$$dg(x_0, x_1)g'(x_2,x_3) \ \mbox{for}\ g, g'\in \{g_1, g_2, g_3\}, d=1,3.$$
The collection of these $\delta$-quadratic terms is denoted by $\delta_Q^{(1)}(4, 4)$. We will continue the study of the quadratic terms $\delta_Q(4, 4)$ in Example 7 by another BH matrix.
\end{example}

\section{Constructions from Hadamard Matrices over Finite Fields}

In this section, we will set $p$  prime. We first introduce some notations in this and the next sections.

Recall that $(x_0,x_1,\cdots,x_{n-1})$ be the $p$-ary expansion of integer $y$, i.e, $y=\sum_{j=0}^{n-1}x_jp^j$. Define the mapping $\gamma$ from $\Z_{p^n}$ to $
\F_{p^n}$ by $$\gamma(y)=\sum_{j=0}^{n-1}x_j\alpha_j,$$
where $\{\alpha_0, \alpha_1, \cdots \alpha_{n-1}\}$ is a  basis of finite field $\F_{p^n}$ over $\F_p$. We will directly use $y\in \F_{p^n}$ instead of $\gamma(y)\in \F_{p^n}$ if the context is clear in this and the next sections. Then we have variables $y_k=\sum_{j=0}^{n-1}\alpha_{j}x_{kn+j}\in \F_{p^n}$ where $\bm{x}_k= (x_{kn}, x_{kn+1}, \cdots, x_{kn+n-1})\in \F_p^n$.

The trace function from $\F_{p^n}$ to $\F_{p}$ is defined by
$$Tr(\alpha)=\alpha+\alpha^p+\alpha^{p^2}+\cdots+\alpha^{p^{n-1}}.$$

\subsection{Constructions by the Trace Function and PPs over $\F_{2^n}$}

In this subsection, we will set $p=2$ and $q$  even. From Corollary 1, the set of the $\delta$-linear terms  is given by
 \begin{equation}\label{linear-2}
\delta_L(q, p^n=2^n)=\left\{ \sum_{k=0}^{m-1}\sum_{i=1}^{2^n-1}\left(c_{k,i}\cdot\prod_{j=0}^{n-1}x_{kn+j}^{i_j}\right)+c' \bigg| c_{k,i}, c'\in \Z_q\right\},
\end{equation}
where  $(i_0,i_1,\cdots,i_{n-1})$ is the binary expansion of integer $i$.

Let $H$ be a Hadamard matrix with entry $H_{u, v}=(-1)^{Tr(u\cdot v)}=w^{\frac{q}{2}Tr(u\cdot v)}$ for $u, v\in \F_{2^n}$. Then the entry of its  phase matrix is given by $\widetilde{H}_{u, v}=\frac{q}{2}Tr(u\cdot v)$. And we have
$$\frac{q}{2}Tr(g(y_0)g'(y_1))\in \delta_Q(q, 2^n)$$ for arbitrary  PPs $g(\cdot), g'(\cdot)$  over $\F_{2^n}$. These terms are called  the  $\delta$-quadratic terms determined by Hadamard matrix over $\F_{2^n}$.

\begin{remark}
$Tr(g(y_0)g'(y_1))$ can be expressed by a Boolean  function from  $\F_2^{2n}$ to $\F_2$  with variables $(\bm{x}_{0}, \bm{x}_{1})$. For $n\geq 2$, since the degree of the PPs $g(y)$ over $\F_2^{2n}$ must be not more than $2^n-2$ with respected to variables $y$, the degree of the $Tr(g(y_0)g'(y_1))$ must be no more than $2n-2$ with respected to Boolean variables $(\bm{x}_{0}, \bm{x}_{1})$.
\end{remark}

\begin{theorem}\label{thm-5}
Let $g_k(\cdot), g_k'(\cdot)$ are arbitrary  PPs over $\F_{2^n}$ for  $0\leq k\leq m$, $\pi$  arbitrary  permutation of symbols $\{0, 1, \cdots mn-1\}$, and $\forall l(\bm{x})\in \delta_L(q, 2^n)$. For any $q$-ary GBF $f(\bm{x})$ from $\F_{2}^{mn}$ to $\Z_q$ with the form
\begin{equation}\label{array-2}
f(\bm{x})=\frac{q}{2}\sum_{k=1}^{m-1}Tr(g_{k}(y_{k-1})\cdot  g'_{k}(y_{k}))+l(\bm{x}),
\end{equation}
 we have the following results.
\begin{itemize}
	\item[(1)] The following $q$-ary  GBFs form a CAS of size $2^n$:
\begin{equation}\label{CSS-2}
f_u(\bm{x})=f(\bm{x})+\frac{q}{2}Tr(u\cdot g_{0}(y_{0})) \ \mbox{for} \ u\in \F_{2^n}.
\end{equation}
Consequently, the $q$-ary sequences  evaluated by GBFs $\{f_{u}(\pi\cdot \bm{x})\mid u\in \F_{2^n}\}$ form a CSS of size $2^n$.
   \item[(2)] The following $q$-ary  GBFs form a CCA:
\begin{equation}\label{CCC-2}
f_{u,v}(\bm{x})=f(\bm{x})+\frac{q}{2}Tr(u\cdot g_{0}(y_{0}))+\frac{q}{2}Tr(v\cdot g'_{0}(y_{m-1})) \ \mbox{for} \ u, v\in \F_{2^n}.
\end{equation}
Consequently, the $q$-ary sequences  evaluated by GBFs $\{f_{u, v}(\pi\cdot \bm{x})\mid u, v\in \F_{2^n}\}$ form a CCC.
\end{itemize}
\end{theorem}

With the same arguments in the subsection 4.1 to avoid the duplication, recall the the semi-normalized PPs in Definition \ref{semi}.

\begin{corollary}\label{coro-5}
For $\forall g_k(\cdot), g_k'(\cdot)\in \mathcal{S}^{(2^n)}$ ($0\leq k\leq m-1$), $d_k\in \F_{2^n}^*$,  $l(\bm{x})\in \delta_L(q, 2^n)$ and permutation $\pi$, the $q$-ary sequence evaluated by GBF $f(\pi \cdot \bm{x})$, where
\begin{equation}\label{seq-2}
f(\bm{x})=\frac{q}{2}\sum_{k=1}^{m-1}Tr(d_k\cdot g_{k}(y_{k-1})\cdot  g'_{k}(y_{k}))+l(\bm{x}),
\end{equation}
lies in a CSS of size $2^n$.
\end{corollary}

\begin{remark}
The collection of the sequences in Corollary \ref{coro-4} is the union of the cosets of the first order Reed-Muller code RM$_q$($1, mn$). Moreover, every sequences in Corollary \ref{coro-4} lies in RM$_q$($2n-2), mn$) for $n\geq 2$.
\end{remark}

We give examples for $n=1$ and $2$ to illustrate the constructions in Theorem \ref{thm-5} and Corollary \ref{coro-5}.

\begin{example}\label{exam-6}
For $n=1$, $\mathcal{S}^{(2)}$ contains only one PP: $g(x)=x$. The functions in Corollary \ref{coro-5} can be expressed by
$$
f(\bm{x})=\frac{q}{2}\sum_{k=1}^{m-1}x_{\pi(k-1)}x_{\pi(k)}+\sum_{k=0}^{m-1}c_{k}x_k+c'
$$
for $c_k, c'\in \Z_q$. From Theorem \ref{thm-5}, the sequences evaluated by
$$\left\{
\begin{aligned}
&f(\bm{x}),\\
&f(\bm{x})+\frac{q}{2}x_{\pi(0)}
\end{aligned}\right.
$$
form a $q$-ary Golay complementary pair. These results coincide with the known $q$-ary standard Golay sequences \cite{DavisJedwab99}. The sequences evaluated by
$$f(\bm{x})\cdot \bm{J_2}+\frac{q}{2}
          \begin{bmatrix}
            0 & x_{\pi(m-1)} \\
            x_{\pi(0)} & x_{\pi(0)}+x_{\pi(m-1)}
          \end{bmatrix}
          $$
form a CCC.
\end{example}

\begin{example}\label{exam-7}
For $n=2$, $q=4$ and $\F_4=\{0, 1, \alpha, \alpha^2=\alpha+1\}$, we have $\{1, \alpha\}$ is a basis of $\F_{2^2}$ over $\F_2$. $\mathcal{S}^{(2^2)}$ contains  two semi-normalized PPs: $g(y)=y$ and $g(y)=y^2$.

We first study the Boolean function $Tr(d\cdot g(y_{0})\cdot  g'(y_{1}))$ where $y_0=x_0+\alpha x_1$, $y_1=x_2+\alpha x_3$, $d\in \F_{2^2}^*$ and $g(\cdot), g'(\cdot)\in\mathcal{S}^{(2^2)}$. Note that $Tr(\beta)=Tr(\beta^2)$ for $\forall \beta \in \F_{2^2}$, then $\delta$-quadratic terms determined by Hadamard Matrix over $\F_{2^2}$ are given by
$$\delta_Q^{(2)}(4, 4)=\{2Tr(y_0y_1), 2Tr(\alpha y_0y_1), 2Tr(\alpha^2 y_0y_1), 2Tr(y_0y_1^2), 2Tr(\alpha y_0y_1^2), 2Tr(\alpha^2 y_0y_1^2)\}.$$
Here we take $Tr(\alpha y_0y_1)$ and $Tr(\alpha^2 y_0y_1^2)$ as two examples to show the calculation process:
\begin{eqnarray*}
Tr(\alpha y_0y_1)&=&Tr(\alpha (x_0+\alpha x_1)(x_2+\alpha x_3))\\
&=&Tr(\alpha x_0x_2+\alpha^2 x_0x_3+ \alpha^2 x_1x_2 + x_1x_3)\\
&=&Tr(\alpha)\cdot x_0x_2+Tr(\alpha^2)\cdot x_0x_3+ Tr(\alpha^2)\cdot x_1x_2 + Tr(1)\cdot x_1x_3\\
&=&x_0x_2+x_0x_3+ x_1x_2,
\end{eqnarray*}
and
\begin{eqnarray*}
Tr(\alpha^2 y_0y_1^2)&=&Tr(\alpha^2 (x_0+\alpha x_1)(x_2+\alpha x_3)^2)\\
&=&Tr(\alpha^2 (x_0+\alpha x_1)(x_2+\alpha^2 x_3))\\
&=&Tr(\alpha^2 x_0x_2+\alpha x_0x_3+ x_1x_2 + \alpha^2x_1x_3)\\
&=&Tr(\alpha^2)\cdot x_0x_2+Tr(\alpha)\cdot x_0x_3+ Tr(1)\cdot x_1x_2 + Tr(\alpha^2)\cdot x_1x_3\\
&=&x_0x_2+x_0x_3+ x_1x_3,
\end{eqnarray*}
By a similar process, we have
$$\left\{
\begin{aligned}
&Tr(y_0y_1)=x_1x_3+x_0x_3+x_1x_2,\\
&Tr(\alpha^2 y_0y_1)=Tr(y_0y_1)+Tr(\alpha y_0y_1)=x_1x_3+x_0x_2,\\
&Tr(y_0y_1^2)=x_0x_3+x_1x_2,\\
&Tr(\alpha y_0y_1^2)=Tr(y_0y_1^2)+Tr(\alpha^2 y_0y_1^2)=x_0x_2+x_1x_3+x_1x_2.
\end{aligned}\right.
$$
The functions in Corollary \ref{coro-5} can be expressed by
\begin{equation}\label{for-example}
f(\bm{x})=\sum_{k=1}^{m-1}h_k(x_{\pi(2 k-2)},x_{\pi(2 k-1)},x_{\pi(2 k)},x_{\pi(2 k+1)} )+\sum_{k=0}^{m-1}e_kx_{\pi(2k)}x_{\pi(2k+1)} +\sum_{k=0}^{2m-1}c_kx_{k}+c'
\end{equation}
for  $h_k(\cdot, \cdot)\in \delta_Q^{(2)}(4, 4)$ and  $e_k, c_k, c'\in \Z_4$.

By applying  Theorems \ref{thm-5},  the sequences evaluated by
$$\left\{
\begin{aligned}
&f(\bm{x}),\\
&f(\bm{x})+2x_{\pi(0)},\\
&f(\bm{x})+2x_{\pi(1)},\\
&f(\bm{x})+2x_{\pi(0)}+2x_{\pi(1)}
\end{aligned}\right.
$$
form a quaternary CSS of size 4.

Recall the $\delta$-quadratic terms determined by DFT matrix of order $4$, which have been shown in Example \ref{exam-5}. There are two equivalent classes of the quaternary BH matrix of order 4. One class is equivalent to the DFT matrix of order $4$, and the other is equivalent to the Hadamard Matrix over $\F_{2^2}$. Then we have
$$\delta_Q(4, 4)=\delta_Q^{(1)}(4, 4)\cup\delta_Q^{(2)}(4, 4),$$
which exactly contains $18+6=24$ $\delta$-quadratic terms. Then the function derived from $\delta_Q(4, 4)$ and $\delta_Q(4, 4)$ can be expressed by $f(\bm{x})$ in the same form (\ref{for-example}) for  $h_k\in \delta_Q(4, 4)$ and  $e_k, c_k, c'\in \Z_4$.

Moreover, The sequences evaluated by
$$\left\{
\begin{aligned}
&f,\\
&f+2x_{\pi(0)},\\
&f+2x_{\pi(1)},\\
&f+2x_{\pi(0)}+2x_{\pi(1)},
\end{aligned}\right.
\left\{
\begin{aligned}
&f,\\
&f+2x_{\pi(0)}+x_{\pi(1)},\\
&f+2x_{\pi(1)},\\
&f+2x_{\pi(0)}+3x_{\pi(1)},
\end{aligned}\right.
\mbox{or}\
\left\{
\begin{aligned}
&f,\\
&f+3x_{\pi(0)}+x_{\pi(1)}+2x_{\pi(0)}x_{\pi(1)},\\
&f+2x_{\pi(0)}+2x_{\pi(1)},\\
&f+x_{\pi(0)}+3x_{\pi(1)}+2x_{\pi(0)}x_{\pi(1)}
\end{aligned}\right.
$$
form a complementary set of size $4$.

\end{example}

The results in Example \ref{exam-7} agree with the  Constructions 3 and 5 in \cite{CCA}. However, all constructions in \cite{CCA} are  based on a Brute-force method, where the computation is very heavy. From Example \ref{exam-7},  we can see that both the  $\delta$-quadratic terms and the constructed sequences can be explicit given based on the algebraic structure of  BH matrices.

\subsection{Constructions by the Trace Function and PPs over $\F_{p^n}$}

In this subsection, we assume $q=p$ prime and $N=p^n$. From Corollary 2, the set of the $\delta$-linear terms  is given by
 \begin{equation}\label{linear-3}
\delta_L(q=p, p^n)=\left\{ \sum_{k=0}^{m-1}\sum_{i=1}^{p^n-1}\left(c_{k,i}\cdot\prod_{j=0}^{n-1}x_{kn+j}^{i_j}\right)+c' \bigg| c_{k,i}, c'\in \F_p\right\}.
\end{equation}
where  $(i_0,i_1,\cdots,i_{n-1})$ is the $p$-ary expansion of integer $i$.

Let $H$ be a Hadamard matrix with entry $H_{u, v}=w^{Tr(u\cdot v)}$ for $u, v\in \F_{p^n}$. Then we have
$$Tr(g(y_0)g'(y_1))\in \delta_Q(p, p^n)$$ 
for arbitrary  PPs $g(\cdot), g'(\cdot)$  over $\F_{p^n}$. These terms are called  the  $\delta$-quadratic terms determined by Hadamard matrix over $\F_{p^n}$.

\begin{theorem}\label{thm-6}
Let $g_k(\cdot), g_k'(\cdot)$ are arbitrary  PPs over $\F_{p^n}$ for  $0\leq k\leq m$, $\pi$  arbitrary  permutation of symbols $\{0, 1, \cdots mn-1\}$, and $\forall l(\bm{x})\in \delta_L(p, p^n)$. For any $p$-ary function $f(\bm{x})$ from $\F_{p}^{mn}$ to $\F_p$ with the form
\begin{equation}\label{array-3}
f(\bm{x})=\sum_{k=1}^{m-1}Tr(g_{k}(y_{k-1})\cdot  g'_{k}(y_{k}))+l(\bm{x}),
\end{equation}
 we have the following results.
\begin{itemize}
	\item[(1)] The following $p$-ary functions form a CAS of size $p^n$:
\begin{equation}\label{CSS-3}
f_u(\bm{x})=f(\bm{x})+Tr(u\cdot g_{0}(y_{0})) \ \mbox{for} \ u\in \F_{p^n}.
\end{equation}
Consequently, the $p$-ary sequences  evaluated by functions $\{f_{u}(\pi\cdot \bm{x})\mid u\in \F_{p^n}\}$ form a CSS of size $p^n$.
   \item[(2)] The following $p$-ary functions form a CCA:
\begin{equation}\label{CCC-3}
f_{u,v}(\bm{x})=f(\bm{x})+Tr(u\cdot g_{0}(y_{0}))+Tr(v\cdot g'_{0}(y_{m-1})) \ \mbox{for} \ u, v\in \F_{p^n}.
\end{equation}
Consequently, the $p$-ary sequences  evaluated by functions $\{f_{u, v}(\pi\cdot \bm{x})\mid u, v\in \F_{p^n}\}$ form a CCC.
\end{itemize}
\end{theorem}

We recall the semi-normalized PPs in definition \ref{semi} to avoid the duplication.
\begin{corollary}\label{coro-6}
For $\forall g_k(\cdot), g_k'(\cdot)\in \mathcal{S}^{(p^n)}$ ($0\leq k\leq m-1$), $d_k\in \F_{p^n}^*$,  $l(\bm{x})\in \delta_L(p, p^n)$ and permutation $\pi$, the $p$-ary sequence evaluated by functions $f(\pi \cdot \bm{x})$, where
\begin{equation}\label{seq-3}
f(\bm{x})=\sum_{k=1}^{m-1}Tr(d_k\cdot g_{k}(y_{k-1})\cdot  g'_{k}(y_{k}))+l(\bm{x}),
\end{equation}
lies in a CSS of size $p^n$.
\end{corollary}

\begin{remark}
Theorem \ref{thm-4} is a  special cases of   Theorem \ref{thm-6}  by restricting $n=1$.
\end{remark}

\section{Constructions from Sequences with 2-Level Autocorrelation}

In this section, we assume $q=p$ prime and $N=p^n$.
Let $\bm{s}$ be a $p$-ary sequence of length $N-1=p^n-1$, given by
$$\bm{s}=(s(0), s(1),\cdots, s(p^n-2)).$$
The {\em periodic auto-correlation} of sequence $\bm{s}$ at shift $\tau$ ($0<\tau< N-1$) is defined by
$$\sum_{i=0}^{N-2}{\omega^{s(i+\tau)-s(i)}},$$
where $i+\tau$  is the summation over $\Z_{N-1}$. We say that $\bm{s}$ has  (ideal) 2-level  autocorrelation if its periodic auto-correlation always equals $-1$ for $0<\tau< N-1$.

\subsection{Trace Representation}

For any $p$-ary sequence $\bm{s}$ of length $p^n-1$, there exists a univariate polynomial function, say $h(y)$ from from $\F_{p^n}$ to $\F_p$, such that
$$s(i)=h(\alpha^i),$$
where $\alpha$ is a primitive element in $\F_{p^n}$. Such a polynomial function $h(y)$ can be represented by the sum of the monomial trace term $Tr(\beta_r y^r)$, where $Tr(y)$ is the trace function from $\F_{p^l}$ to $\F_{p}$ and $\beta_r \in \F_{p^l}$ for $l$ being the  coset size of $r$ and $l|n$, i.e., 
 $$h(y) = \sum_{r} Tr(\beta_r y^r),$$
 where $r$'s are coset leaders modulo $p^n-1$.
 
Since the sequence $\bm{s}$ has period $p^n-1$, there is at least one $r$ such that the coset containing $r$ has the full length $l=n$. Moreover,  if $h(y)$ has only one trace term, it is an $m$-sequence, which has been studied in the literature for more than seven decades. The reader is referring to \cite{Gong-book} for more details on the trace representation of sequences with period $p^n-1$.

It is obvious that $\bm{s}$ is a  2-level autocorrelation sequence if and only if its trace representation $h(y)$ satisfies
\begin{equation}\label{orth}
\sum_{y\in \F_{p^n}}\omega^{h(\lambda y)-h(y)}=0 \ \mbox{for}\ \forall \lambda\in \F_{p^n}^*.
\end{equation}

\subsection{ 2-Level Autocorrelation Sequences and BH matrices}

For a given 2-level autocorrelation sequence $\bm{s}$  of period $p^n-1$, we can construct a BH matrix of order $N=p^n$, say $\bm{H}=(H_{ij})$, as follows
$$
\begin{array}
	{l}
H_{i+1, j+1}=\omega^{s(i+j)},  0\le i, j<p^n-1,\\
H_{0,j}=H_{i,0}=1, 0\le i, j <p^n,
\end{array}
$$
where $i+j$  is the summation over $\Z_{N-1}$.

On the other hand,  the BH matrix $\bm{H}$ determined by the  2-level autocorrelation sequence $\bm{s}$,  in the sense of equivalence, can be represented by its  trace representation $h(y)$ with entry
$$H_{u,v}=\omega^{h(u\cdot v)} \ \mbox{for}\ u,v\in \F_{p^n}.$$

\begin{remark}
Note that $h(y)$ is the trace representation of sequence $\bm{s}$, we always have $h(0)=0$, which leads to $H_{0,v}=H_{u,0}=1$.
\end{remark}

\subsection{Known Constructions on 2-level Autocorrelation Sequences}

All the known construction on binary  2-level  autocorrelation sequences are collected in \cite{Gong-book} which remains the record until now. We provide  a summary for those constructions for both binary and nonbinary cases of length $p^n-1$ in the  following outlines.

Binary case ($p=2$):
\begin{enumerate}
\item[(1)] $m$-sequences (Golomb in 1954).
\item[(2)] For Mersenne prime $2^n-1$, quadratic residue sequences (1932).
For Mersenne prime $2^n-1 = 4a^2 + 27$, Hall's sextic residue sequences.
\item[(3)]  For $n \ge  6$, $n$ composite, GMW sequences (Goldon, Mills and Welch \cite{GMW62} in 1962, Scholtz and Welch \cite{SW84} in 1984).
\item[(4)]  Hyper-oval construction: Segre case and Glynn I and II cases (Maschietti \cite{Maschietti} in 1998).
\item[(5)] Dillon-Dobbertin's Kasami power function construction \cite{DD04} including conjectured 3-term and 5-term  sequences \cite{five} as subclasses, and also proved the case of the WG sequences \cite{WG98}.
\end{enumerate}

Nonbinary case:
\begin{enumerate}
\item[(1)] For $p>2$, $m$-sequences (Zieler, 1959), GMW sequences and HG sequences \cite{HG}.
	\item[(2)] For $p=3$,
	\begin{enumerate}
		\item Lin conjectured sequences (Hu, et al. \cite{HGH14} in 2014, Arasu et al. \cite{Arasu15} in  2015).
		\item Conjectured sequences by Ludkovski and Gong \cite{LG01} in 2000, some cases are proved in \cite{Arasu15}.
	\end{enumerate}
\end{enumerate}

For both binary and nonbinary cases, we also have  the subfield constructions: if $1<l<n$ and $l|n$,  $h_1(y)$ is a function from $\F_{p^l}$ to $F_{p}$ whose evaluation has 2-level autocorrelation of length $p^l-1$, and $h_2(y)$ is a GMW function from $\F_{p^n}$ to $\F_{p^l}$, then the composition of $h_1$ and $h_2$ produces a 2-level autocorrelation sequence of length $p^n-1$.

\subsection{$\delta$-Quadratic Terms from Sequences with 2-Level Autocorrelation}

Since a $p$-ary 2-level autocorrelation sequence $\bm{s}$ of period $p^n-1$ can be represented by a trace representation $h(y)$, which determine a BH matrix with entry $H_{u,v}=\omega^{h(u\cdot v)}$. According to Theorem \ref{theorem-q}, new $\delta$-quadratic terms are obtained.

\begin{corollary}
Let $h(y)$ be the trace representation of a $p$-ary sequence with 2-level autocorrelation of period $p^n-1$. Then we have
$$h(g(y_0)\cdot g'(y_1))\in \delta_Q(q=p, p^n),$$
where $g(\cdot), g'(\cdot)$ are arbitrary PPs over $\F_{p^n}$.
\end{corollary}

\begin{corollary}
Let $h(y)$ be the trace representation of a binary sequence with 2-level autocorrelation of period $2^n-1$. Then we have
$$\frac{q}{2}h(g(y_0)\cdot g'(y_1))\in \delta_Q(q, 2^n),$$
where $g(\cdot), g'(\cdot)$ are arbitrary PPs over $\F_{2^n}$, and $q$ is even.
\end{corollary}

\begin{example}
For $m$ sequence, it trace representation is given by $h(y)=Tr(y)$. Then the entry of BH matrix $\bm{H}$ determined by $m$ sequence is given by $H_{u,v}=\omega^{Tr(u\cdot v)}$, which is the Hadamard matrix over $\F_{p^n}$ shown in Section 5. Thus, the results in Section 4.1 and Section 5 can be explained from the viewpoint of the $m$ sequences.
\end{example}

Other constructions of 2-level autocorrelation sequences yield new $\delta$-quadratic terms and new constructions of CSSs and CCCs. We give 3-term  sequences to illustrate it.

\begin{example}
For $p=2$, odd $n\geq 5$, $n=2n'+1$,  the binary 3-term sequence
$$s_i=Tr(\alpha^i)+Tr(\alpha^{q_1i})+Tr(\alpha^{q_2i})$$
has 2-level auto-correlation, where $q_1=2^{n'}+1$ and  $q_1=2^{n'}+2^{n'-1}+1$. Its trace representation is given by
$$h(y)=Tr(y+y^{q_1}+y^{q_2}).$$
Then the entry of BH matrix $\bm{H}$ determined by this three-term sequence is given by
$$H_{u,v}=(-1)^{Tr(uv+(uv)^{q_1}+(uv)^{q_2})}.$$
And we have
$$\frac{q}{2}Tr(g(y_0)\cdot g'(y_1)+g(y_0)^{q_1}\cdot g'(y_1)^{q_1}+g(y_0)^{q_2}\cdot g'(y_1)^{q_2})\in \delta_Q(q, 2^n),$$
where $g(\cdot), g'(\cdot)$ are arbitrary PPs over $\F_{2^n}$, and $q$ is even.

In particular, if we set $n=5$, we have that $\bm{H}$ with entry $H_{u,v}=(-1)^{Tr(uv+(uv)^{5}+(uv)^{7})}$ is a binary BH matrix of order 32, and 
$$\frac{q}{2}Tr(g(y_0)\cdot g'(y_1)+g(y_0)^{5}\cdot g'(y_1)^{5}+g(y_0)^{7}\cdot g'(y_1)^{7})\in \delta_Q(q, 2^5).$$
\end{example}

\section{Concluding Remarks}

The theory for $\delta$-quadratic functions in this paper is for arbitrary BH matrices, though we only discuss some special cases such as DFT matrices and BH matrices derived from 2-level autocorrelation sequences. On the other hand, even for the binary case, there are 5 inequivalent BH matrices of order 16, millions of inequivalent BH matrices of orders 32. The future study on these inequivalent BH matrices will produce new CSSs and CCCs, and can exponentially increase the number of sequences with low PMEPR.


\end{document}